\documentclass[final,12pt]{article}
\usepackage{amsfonts,color,morefloats,pslatex,a4wide}
\usepackage{amssymb,amsthm,amsmath,latexsym,pslatex,cite}
\usepackage{lineno,hyperref,mathtools}
\usepackage{pdflscape}

\newtheorem{theorem}{Theorem}[section]

\newtheorem{lemma}[theorem]{Lemma}
\newtheorem{corollary}[theorem]{Corollary}
\newtheorem{proposition}[theorem]{Proposition}

\theoremstyle{definition}
\newtheorem{definition}[theorem]{Definition}
\newtheorem{example}[theorem]{Example}

\theoremstyle{remark}
\newtheorem{remark}[theorem]{Remark}

\newcommand{\F}{\mathbb{F}}
\newcommand{\x}{\mathbf{x}}

\newcommand{\C}{{\mathcal{C}}}

\newcommand{\be}{\begin{eqnarray}}
\newcommand{\ee}{\end{eqnarray}}
\newcommand{\nn}{{\nonumber}}

\newcommand{\ra}{\rightarrow}

\makeatletter

\newcommand{\Rmnum}[1]{\expandafter\@slowromancap\romannumeral #1@}
\makeatother

\begin{document}
\begin{sloppypar}

\title{Determining the covering radius of all generalized Zetterberg codes in odd characteristic\thanks{The research of Minjia Shi and Shitao Li is supported by the National Natural Science Foundation of China under Grant 12471490. The research of Tor Helleseth was supported by the Research Council of Norway under Grant 247742/O70. The research of Ferruh \"{O}zbudak is supported by T\"{U}B\.{I}TAK under Grant 223N065.}}

 \author{Minjia Shi\thanks{Minjia Shi is with the Key Laboratory of Intelligent Computing and Signal Processing, Ministry of Education, School of Mathematical Sciences, Anhui University, Hefei, 230601, China, (email: smjwcl.good@163.com).},~
 Shitao Li\thanks{Shitao Li is with the School of Mathematical Sciences, Anhui University, Hefei, 230601, Anhui, China, (email: lishitao0216@163.com).},~
 Tor Helleseth\thanks{Tor Helleseth is with the Department of Informatics, University of Bergen, Bergen, Norway, (email: tor.helleseth@uib.no).},~
 Ferruh \"{O}zbudak\thanks{Ferruh \"{O}zbudak is with the Faculty of Engineering and Natural Sciences, Sabanc{\i} University, 34956 Istanbul, Turkiye, (email: ferruh.ozbudak@sabanciuniv.edu).} }

\maketitle

\begin{abstract}
For an integer $s\ge 1$, let $\mathcal{C}_s(q_0)$ be the generalized Zetterberg code of length $q_0^s+1$ over the finite field $\F_{q_0}$ of odd characteristic. Recently, Shi, Helleseth, and \"{O}zbudak (IEEE Trans. Inf. Theory 69(11): 7025-7048, 2023) determined the covering radius of $\mathcal{C}_s(q_0)$ for $q_0^s \not \equiv 7 \pmod{8}$, and left the remaining case as an open problem. In this paper, we develop a general technique involving arithmetic of finite fields and algebraic curves over finite fields to determine the covering radius of all generalized Zetterberg codes for $q_0^s \equiv 7 \pmod{8}$, which therefore solves this open problem. We also introduce the concept of twisted half generalized Zetterberg codes of length $\frac{q_0^s+1}{2}$, and show the same results hold for them. As a result, we obtain some quasi-perfect codes.
\end{abstract}

\section{Introduction}
Let $\F_{q_0}$ denote a finite field with $q_0$ elements, where $q_0$ is a prime power. Let $q=q_0^s$ for an integer $s \ge 1$ and let
$\F_{q_0^s}^*$ denote the multiplicative group of the field extension $\F_{q_0^s}$, so that $\F_{q_0^s}^*=\F_{q_0^s} \setminus \{0\}$.
The {\em (Hamming) weight} of a vector ${\bf c}\in \F_{q_0}^n$, denoted by $wt({\bf c})$, is the number of nonzero components in ${\bf c}$. A ${q_0}$-ary {\em linear $[n,k,d]$ code} $\C$ is a $k$-dimensional subspace of $\mathbb F_{{q_0}}^n$, where $d$ is the minimum weight of all nonzero codewords of $\C$. The {\em (Euclidean) dual} code $\C^{\perp}$ of a ${q_0}$-ary linear code $\C$ of length $n$ is defined by
$\C^{\perp}=\{\textbf y\in \F_{{q_0}}^{n}~|~\langle \textbf x, \textbf y\rangle=0~ {\rm for\ all}\ \textbf x\in \C \},$
where $\langle \textbf x, \textbf y\rangle=\sum_{i=1}^n x_iy_i$ for ${\bf x} = (x_1, x_2, \ldots, x_n)$ and $\textbf y = (y_1, y_2, \ldots, y_n)\in \F_{{q_0}}^{n}$.
A {\em parity check matrix} of a $q_0$-ary linear $[n, k]$ code $\C$ is any $(n-k)\times n$ matrix $P$ whose rows form a basis of $\C^\perp$.

Let $\C$ be a ${q_0}$-ary linear $[n,k,d]$ code with parity check matrix $P$. The {\em packing radius} $t(\C)$ of $\C$ is defined by $t(\C)=\left\lfloor\frac{d-1}{2}\right\rfloor$. The {\em covering radius} of $\C$, denoted by $\rho(\C)$, is the smallest integer such that the balls of this radius around the codewords cover the complete space $\F_{q_0}^n$, equivalently, is the least integer $\rho$ such that every column vector of $n-k$ entries is a sum of some $\rho$ or fewer columns of $P$.
The covering radius is a basic geometric parameter of a code. It has various applications including decoding, data compression, testing, write-once memories and combinatorics in general. For further details on the significance and applications of covering radius of codes,
we refer, for example, to \cite{Brualdi-Litsyn-Pless,Cohen-book,CKMS-IT,CLLM-AAECC}, and the references therein.

There are important connections to perfect codes and quasi-perfect codes. It is well-known that for any code $C$, its packing radius $t(C)$ is at most its covering radius $\rho(C)$, and $C$ is called {\em perfect} if $t(C)= \rho(C)$, and {\em quasi-perfect} if $t(C)+1= \rho(C)$. In the 1970's, the parameters of perfect codes over finite fields were completely determined \cite{perfect1,perfect2,perfect3}. The classification task for the possible parameters of quasi-perfect codes is much more complicated. A more achievable objective is to construct quasi-perfect codes, which has been studied for many years (see \cite{ternary-QP,ternary-QP2,GPZ,LH-CCDS2016,Z,SM-1,SM,Survey-IT,Helleseth-IT}). There is also an important connections to complete caps in projective spaces. There is a wide literature on complete caps, we refer, for example, \cite{Bruen-W,cossidente,Davydov-M-P,Giuliett,Hirschfeld-Storme}.

However, the problem of finding the covering radius of a given code is very hard in general \cite{DJ-MC}. It has been proven in \cite{M-IT1984} that computing the exact covering radius of a code is both NP-hard and co-NP-hard. In general, most of the results give bounds on the covering radii rather than exact values \cite{AB-DCC2002,Helleseth-DAM,Sole-IT1995,T-DAM1987,MC-IT2003,JCTA1}. There are only a few classes of codes with special parameters in which the covering radii are known. A non-exhaustive list is \cite{ternary-QP,ternary-QP2,GPZ,LH-CCDS2016,SM-1,SM,Helleseth-IT,JCTA2}. Recently, Shi {\em et al.} \cite{SHOS2022IT} develop some interesting techniques to determine the exact covering radii of Melas codes, which supplemented the results of \cite{SM-1,M1,MC-IT2003,VB}. Very recently, Shi, Helleseth, and \"{O}zbudak \cite{SHO2023IT} presented a generalization of binary Zetterberg codes \cite{Z}, proposed the concept of generalized Zetterberg codes, and further studied the exact coverage radii of such codes.
Specifically, let $s \ge 1$ be an integer. Put $q=q_0^s$ and $n=q+1$. Let $H$ be the subgroup of $\F_{q^2}^*$ with $|H|=n$.
Let $\{h_1, \ldots, h_n\}$ be an enumeration of $H$. The generalized Zetterberg code $\mathcal{C}_s(q_0)$
of length $n=q_0^s+1$ over $\F_{q_0}$ is the $\F_{q_0}$-linear code  with the parity check matrix
\begin{align}\label{def.Zetterberg}
P=\left[ h_1 \; h_2 \; \cdots h_n\right].
\end{align}
In fact, each column $h_i$ in $P$ is considered as $\phi(h_j) \in \F_{q_0}^{2s}$, where $\phi$ is any $\F_{q_0}$-linear bijective map from $\F_{{q^2}}$ to $\F_{q_0}^{2s}$.
It has been proven that $\mathcal{C}_s(q_0)$ has dimension $n-2s$ (see \cite[Lemma A.1]{SHO2023IT}). Note that $H$ is omitted in the notation $\mathcal{C}_s(q_0)$ since $H$ is the unique subgroup of $\F_{q^2}^*$ of order $n$. The same paper \cite{SHO2023IT} determined the covering radius of $\mathcal{C}_s(q_0)$ for $q_0^s \not \equiv 7 \pmod{8}$, and left the remaining case as an open problem. In fact, they \cite{SHO2023IT} also proposed the concepts of half and twisted half generalized Zetterberg codes in odd characteristic, show that their covering radii are the same as the covering radii of the (full) generalized Zetterberg codes, and further utilized them to obtain some quasi-perfect codes over different finite fields.

Please note that if $s$ is even, then $q_0^s \not\equiv 7 \pmod{8}$, this case have been completely solved in \cite{SHO2023IT}. Therefore, we will only consider that $q_0^s \equiv 7 \pmod{8}$ and $s$ is odd in this paper. We develop a general technique involving arithmetic of finite fields and algebraic curves over finite fields to determine the covering radius of all generalized Zetterberg codes for $q_0^s \equiv 7 \pmod{8}$, which therefore solves the above open problem proposed in \cite{SHO2023IT}. Specifically, for $q_0^s \equiv 7 \pmod{8}$ and $s\geq 3$, we present an explicit bound $s^*$ so that if $s\geq s^*$, then the covering radius of $\mathcal{C}_s(q_0)$ is 3. we also show that the covering radius of $\mathcal{C}_s(q_0)$ is 2 for $s=1$. We also introduce the concept of twisted half generalized Zetterberg codes of length $\frac{q_0^s+1}{2}$, and show the same results hold for them. As a result, we obtain some quasi-perfect codes.

This paper is organized as follows. Next section provides some useful results related to this paper. In Section \ref{section.upper bounds}, we prove the covering radius of $\mathcal{C}_s(q_0)$ is at most $3$, which is a long and quite technical section. In Section \ref{section4}, we determine the exact covering radius in most cases. In Section \ref{section.half twisted half}, we extend our results to twisted half generalized Zetterberg codes. We conclude in Section \ref{section6}.

\section{Preliminaries}

Let $\F_{q_0}$ be a finite field of odd characteristic. Let $q=q_0^s$ for an integer $s \ge 1$. Let
$$S=\{\F_q:~q~\mbox{is odd and}~q \not \equiv 7\!\!\!\!\!\pmod{8}\}.$$
Note that $\F_7, \F_{47}, \F_{31} \not \in S$.
In \cite{SHO2023IT}, $\mathcal{C}_s(q_0)$ of length $q+1$ and is considered for all $q \in S$.
For integers $\ell \ge 3$, let
\be
S_{\ell}=\{ \F_q: \; \mbox{$q$ is odd and $q  \equiv 2^\ell-1\!\!\!\!\! \pmod{2^{\ell +1}}$}\}.
\nn\ee

In this section, $\mathcal{C}_s(q_0)$ of length $q+1$ is considered for all integers $\ell \ge 3$  and  all $q \in S_\ell$.
In the next lemma we show that the set
\be
S \cup \bigcup_{\ell=3}^\infty S_\ell
\nn\ee
is the set of all finite fields of odd characteristic.
For example, $\F_7 \in S_3$, $\F_{47} \in S_4$ and $\F_{31}  \in S_5$.

In particular, this means that we consider $\mathcal{C}_s(q_0)$ of length $q+1$ for all finite fields of odd characteristic, which completes the study started in \cite{SHO2023IT} by solving the corresponding open problem.

\begin{lemma} \label{lemma1.determineCR}
Let $\F_q$ be a finite field of odd characteristic. If $q \equiv 7 \pmod{8}$, then there  exists an integer $ \ell \ge 3$ such that
\be
q \equiv 2^\ell -1\!\!\!\!\! \pmod{2^{\ell +1}}.
\nn\ee
\end{lemma}
\begin{proof}
As $q \equiv 7 \pmod{8}$, it is clear that $q \ge 7$. Let $t$ be the smallest positive integer such that
$q< 2^t.$
We consider the $2$-adic expansion of $q$ given by
\be
q=1+2+4+a_3 2^3 + a_4 2^4 + \cdots +a_{t-1} 2^{t-1},
\nn\ee
where $a_3, a_4, \ldots, a_{t-1} \in \{0,1\}$. Here we use the fact that $q \equiv 7 \pmod{8}$ as $a_0=a_1=a_2=1$ in the expansion above.
First assume that
$(a_3,a_4, \ldots, a_{t-1})=(1,1, \ldots, 1).$
Put $\ell=t$. It is clear that
\be
q=1+2+ \cdots + 2^{t-1}=2^t-1 \equiv 2^\ell -1 \!\!\!\!\! \pmod{2^{\ell +1}}).
\nn\ee

Next we assume that there exists zero in $(a_3,a_4, \ldots, a_{t-1})$. Let $\ell$ be the smallest index $3 \le \ell \le  t-1$ such that $a_\ell =0$. Then we have
\be
q=1+2+4+ \cdots + 2^{\ell -1} + 0 + a_{\ell +1} 2^{\ell +1} + a_{\ell +2} 2^{\ell +2} + \cdots + a_{t-1}2^{t-1}.
\nn\ee
This implies that
\be
q \equiv 1 + 2 + 4 + \cdots + 2^{\ell -1} \equiv  2^{\ell} -1 \!\!\!\!\! \pmod{2^{\ell +1}}
\nn\ee
independent from the values of $a_{\ell +1}, a_{\ell +2}, \cdots, a_{t-1}$. This completes the proof.
\end{proof}

Let $\F_{q_0}$ be a finite field of odd characteristic and $s \ge 1$ be an integer. Put $q=q_0^s$. Note that
\be
\F_q \in S \rlap{$\quad\not$}\implies \F_{q_0} \in S.
\nn\ee
For example if $q_0=7$ and $s=2$, then we have
\be
\F_{7^2} \in S ~ \mbox{but}~ \F_7 \not \in S.
\nn\ee

As a consequence of this situation, in \cite{SHO2023IT} the cases $s$ even and $s$ odd are considered separately.
However if we change $S$ with $S_3$, then this is not true. Let $\F_{q_0}$ be a finite field of odd characteristic and $s \ge 1$ be an integer. Put $q=q_0^s$.
We have
\be \label{cond.S3}
\F_q \in S_3 \iff \F_{q_0} \in S_3 \; \mbox{and} \; \mbox{$s$ is odd}.
\ee
Indeed let $s \ge 1$ be an integer. Assume that $q_0 \equiv 7 \pmod{16}$. Then $q_0^2 \equiv 49 \equiv 1 \pmod{16}$ and hence we have
\be
q_0^s \equiv 7\!\!\!\!\! \pmod{16} \iff \mbox{$s$ is odd}.
\nn\ee
This completes the proof of the sufficiency of the statement in Eq. (\ref{cond.S3}).
Moreover we have
\be
\begin{array}{l}
q_0 \equiv 1 \!\!\!\!\!\pmod{16} \implies q_0^s \equiv 1 \!\!\!\!\!\pmod{16}, \\
q_0 \equiv 3 \!\!\!\!\!\pmod{16} \implies q_0^s \equiv a \!\!\!\!\!\pmod{16} \; \mbox{with} \; a \in \{3,9,11,1\}, \\
q_0 \equiv 5 \!\!\!\!\!\pmod{16} \implies q_0^s \equiv a \!\!\!\!\!\pmod{16} \; \mbox{with} \; a \in \{5,9,13,1\}, \\
q_0 \equiv 7  \!\!\!\!\!\pmod{16} \implies q_0^s \equiv a \!\!\!\!\!\pmod{16} \; \mbox{with} \; a \in \{7,1\}, \\
q_0 \equiv 9 \!\!\!\!\!\pmod{16} \implies q_0^s \equiv a \!\!\!\!\!\pmod{16} \; \mbox{with} \; a \in \{9,1\}, \\
q_0 \equiv 11 \!\!\!\!\!\pmod{16} \implies q_0^s \equiv a \!\!\!\!\!\pmod{16} \; \mbox{with} \; a \in \{11,9,3,1\}, \\
q_0 \equiv 13 \!\!\!\!\!\pmod{16} \implies q_0^s \equiv a \!\!\!\!\!\pmod{16} \; \mbox{with} \; a \in \{13,9,5,1\}, \\
q_0 \equiv 15 \!\!\!\!\!\pmod{16} \implies q_0^s \equiv a \!\!\!\!\!\pmod{16} \; \mbox{with} \; a \in \{15,1\}. \\
\end{array}
\nn\ee
These arguments imply that the necessity statement of Eq. (\ref{cond.S3}) holds as well.

Consequently if $\F_{q} \in S_3$ and $q=q_0^s$ for an integer $s \ge 1$, then we can assume that $\F_{q_0} \in S_3$ and $s$ is odd without loss of generality.

In the next two lemmas we generalize Eq. (\ref{cond.S3}) for any integer $\ell \ge 3$. Namely, in particular, we prove the following: Let $\ell \ge 3$ be an integer. If $\F_{q} \in S_\ell$ and $q=q_0^s$ for an integer $s \ge 1$, then we can assume that $\F_{q_0} \in S_\ell$ and $s$ is odd without loss of generality.

We use induction in the proof of the next lemma. Hence it is more suitable to state it for an odd integer  $A\ge 1$. We will apply it for a finite field $\F_{q_0}$
of odd characteristic such that $A=q_0$.

\begin{lemma} \label{lemma2.determineCR}
Let $A \ge 1$ be an odd integer.
Let $\ell \ge 2$ be an integer. If there exists an integer $s \ge 1$ such that
\be\label{e0.lemma2.determineCR}
A^s \equiv 2^\ell -1 \!\!\!\!\! \pmod{2^{\ell}},
\ee
then we have
\begin{itemize}
\item[{\rm 1)}] $A \equiv 2^\ell -1  \pmod{2^{\ell}}$, and
\item[{\rm 2)}] $s$ is odd.
\end{itemize}
Conversely if the items {\rm 1)} and {\rm 2)} above hold, then Eq. (\ref{e0.lemma2.determineCR}) holds.
\end{lemma}
\begin{proof}
The converse statement is clear. We use induction on $\ell$ for the necessity part. The proof is clear if $\ell=2$. First assume that the step $\ell$ holds by induction hypothesis. Regarding the step $\ell +1$, assume that there exists an integer $s \ge 1$ such that
\be \label{ep0.lemma2.determineCR}
A^s \equiv 2^{\ell +1} -1 \!\!\!\!\! \pmod{2^{\ell+1}}.
\ee
Let $A_0$ and $A_1$ be the integers such that
\be \label{ep1.lemma2.determineCR}
A \equiv A_0 + 2^\ell A_1 \!\!\!\!\! \pmod{2^{\ell+1}} \;\; \mbox{with} \;\; 0 \le A_0 \le 2^{\ell -1} \;\; \mbox{and} \;\; 0 \le A_1 \le 1.
\ee
Note that $A_0$ is an odd integer. Combining Eq. (\ref{ep0.lemma2.determineCR}), Eq. (\ref{ep1.lemma2.determineCR}) and considering modulo $2^\ell$, we obtain
\be \label{ep2.lemma2.determineCR}
A_0^s \equiv (A_0 + 2^\ell A_1)^s \equiv -1 \equiv 2^\ell -1  \!\!\!\!\! \pmod{2^{\ell}}.
\ee
Applying the induction hypothesis of step $\ell$ for $A_0$ (note that $A_0$ is a positive odd integer) in Eq. (\ref{ep2.lemma2.determineCR}) we obtain that $A_0 \equiv 2^\ell -1 \pmod{2^{\ell}}$ and $s$ is odd. It remains to show that $A_1 \neq 0$. Assume the contrary that $A_1=0$, which implies that
\be \label{ep3.lemma2.determineCR}
A \equiv 2^\ell -1 \!\!\!\!\! \pmod{2^{\ell+1}}
\ee
and
\be \label{ep4.lemma2.determineCR}
A^2 \equiv (2^\ell -1)^2 \equiv 2^{2\ell} -2^{\ell +1} +1 \equiv 1 \!\!\!\!\! \pmod{2^{\ell+1}}.
\ee
Combining Eq. (\ref{ep3.lemma2.determineCR}) and Eq. (\ref{ep4.lemma2.determineCR}) we conclude that $A^s \not \equiv 2^{\ell +1} -1  \pmod{2^{\ell+1}}$, which is a contradiction to Eq. (\ref{ep0.lemma2.determineCR}). Hence $A_1=1$ and we have
\be
A \equiv 2^\ell -1 + 2^\ell \equiv 2^{\ell+1} -1 \!\!\!\!\! \pmod{2^{\ell+1}},
\nn\ee
which completes the proof.
\end{proof}

Using Lemma \ref{lemma2.determineCR} we prove the following lemma.

\begin{lemma} \label{lemma3.determineCR}
Let $A \ge 1$ be an odd integer.
Let $\ell \ge 2$ be an integer. If there exists an integer $s \ge 1$ such that
\be\label{e0.lemma3.determineCR}
A^s \equiv 2^\ell -1 \!\!\!\!\! \pmod{2^{\ell+1}},
\ee
then we have
\begin{itemize}
\item[{\rm 1)}] $A \equiv 2^\ell -1 \pmod{2^{\ell+1}}$, and
\item[{\rm 2)}] $s$ is odd.
\end{itemize}
Conversely if the items {\rm 1)} and {\rm 2)} above hold, then Eq. (\ref{e0.lemma3.determineCR}) holds.
\end{lemma}

\begin{proof}
The converse statement is clear. We prove the necessity part using Lemma \ref{lemma2.determineCR}. Using Lemma \ref{lemma2.determineCR} and considering Eq. (\ref{e0.lemma3.determineCR}) modulo $2^\ell$ we obtain that $s$ is odd and
\be
A \equiv 2^\ell -1 \!\!\!\!\! \pmod{2^{\ell}}.
\nn\ee
Let $0 \le A_1 \le 1$ be the integer such that
\be
A \equiv 2^\ell -1 + 2^\ell A_1 \!\!\!\!\! \pmod{2^{\ell+1}}.
\nn\ee
Next we show that $A_1=0$. Assume the contrary that $A_1=1$, which implies that 
 \begin{align}\label{ep1-determineCR}
 A \equiv 2^\ell -1 + 2^\ell \equiv -1 \!\!\!\!\! \pmod{2^{\ell+1}},
 \end{align}
and
\begin{align}\label{ep2-determineCR}
    A^2 \equiv 1 \!\!\!\!\! \pmod{2^{\ell+1}}.
   \end{align}

Combining Eq. (\ref{ep1-determineCR}) and Eq. (\ref{ep2-determineCR}), we obtain that $A^s \not \equiv 2^\ell -1 \pmod{2^{\ell+1}}$ for any integer $s \ge 1$, which is a contradiction to Eq. (\ref{e0.lemma3.determineCR}). This implies that $A_1=0$ and
$$A\equiv 2^\ell -1 \!\!\!\!\! \pmod{2^{\ell+1}},$$
which completes the proof.
\end{proof}

Let $\ell \ge 3$ be an integer. Using Lemmas \ref{lemma1.determineCR}, \ref{lemma2.determineCR} and \ref{lemma3.determineCR}, we assume that
\be \label{Assumption1}
q_0 \equiv 2^\ell -1 \!\!\!\!\! \pmod{2^{\ell+1}}
\ee
and $s \ge 1$ is an odd integer with
\be \label{Assumption2}
q=q_0^s
\ee
{\it without loss of generality}.  We note that the phrase ``without loss of generality" here refers to Lemmas \ref{lemma1.determineCR}, \ref{lemma2.determineCR} and \ref{lemma3.determineCR}.

\section{Covering radius is at most $3$} \label{section.upper bounds}
Recall that $\ell \ge 3$ is an integer. Using Eq. (\ref{Assumption1}) and Eq. (\ref{Assumption2}) we obtain that
\be \label{Assumption3}
q\equiv 2^\ell -1 \!\!\!\!\! \pmod{2^{\ell+1}}.
\ee
Let $\theta$ be a primitive $2^\ell$-th root of $1$ in $\F_{q^2}^*$. Using Eq. (\ref{Assumption3}) we get that $\theta \in \F_{q^2} \setminus \F_q$.

Let $H$ be the multiplicative subgroup of $\F_{q^2}^*$ with $|H|=q+1$. Put $m=(q_0-1)/2$. Let $H_m$ be the multiplicative subgroup of $\F_{q^2}^*$ with $|H_m|=m(q+1)$.

For index $0 \le i \le 2^\ell -1$, let Property Pi be the property defined as follows:
\begin{definition}
Property Pi: For each $\gamma \in \F_{q^2}^*$ such that $\gamma^q=\theta^i \gamma$, there exist $h_1,h_2,h_3 \in H_m$ such that
\be
h_1+h_2 + h_3=\gamma.
\nn\ee
\end{definition}

Recall that the covering radius of a linear code is the least integer $\rho$ such that every column vector of $n-k$ entries is a sum of some $\rho$ or fewer columns of its parity check matrix.

\begin{theorem} \label{theorem.upper1}
Let $\ell \ge 3$ be an integer. Let $\F_{q_0}$ be a finite field such that Eq. (\ref{Assumption1}) holds. Let $s \ge 1$ be an odd integer and put $q=q_0^s$. Then the covering radius of $\mathcal{C}_s(q_0)$ is at most $3$ if Property Pi holds for each index $0 \le i \le 2^{\ell} -1$.
\end{theorem}
\begin{proof}
Note that the covering radius of $\mathcal{C}_s(q_0)$ is at most $3$ if for each given $\mu \in \F_{q^2}$, there exist $c_1,c_2,c_3 \in \F_{q_0}^*$ and $h_1,h_2,h_3 \in H$ such that
\be
c_1 h_1 + c_2h_2 + c_3 h_3 = \mu.
\nn\ee

If $\mu=0$, then this is trivial. Hence we assume that $\mu \in \F_{q^2}^*$ without loss of generality.
For $0 \le i \le 2^\ell -1$, let $G_i$ be the subset defined as
\be
G_i=\{\gamma \in \F_{q^2}^*: \gamma^q = \theta^i \gamma\}.
\nn\ee
Note that $|G_i|=q-1 \;\; \mbox{and} \;\; G_i \cap G_j=\emptyset$
for $0 \le i  < j  \le 2^\ell -1$.
Let $G \subseteq \F_{q^2}^*$ be the disjoint union
\be
G=\bigcup_{i=0}^{2^\ell -1} G_i.
\nn\ee
We observe that $G$ is a multiplicative subgroup of $\F_{q^2}^*$ and
$|G|=2^\ell (q-1).$
Next we show that
\be
\gcd(|G|,|H|)= 2^\ell.
\nn\ee
Indeed we observe that $\frac{q+1}{2^\ell}$ is an integer by the assumption in Eq. (\ref{Assumption3}) and we have $2=2^\ell \frac{q+1}{2^\ell} -(q-1)$. Hence $\gcd(\frac{q+1}{2^\ell},q-1)$ is either $1$ or $2$. Moreover $2 \nmid \frac{q+1}{2^\ell}$ as $q +1\equiv 2^\ell \pmod{2^{\ell+1}}$. These arguments imply that $\gcd(\frac{q+1}{2^\ell},q-1)=1$ and
\be
\gcd(|G|,|H|)=2^\ell \gcd\left(q-1, \frac{q+1}{2^\ell}\right)=2^\ell.
\nn\ee
Therefore we obtain that
\be
{\rm lcm}(|G|,|H|)=\frac{2^\ell (q-1)(q+1)}{2^\ell}=q^2-1.
\nn\ee

These arguments imply that if $\mu \in \F_{q^2}^*$, then there exists a unique index $0 \le i \le 2^\ell -1$ such that $\mu=\gamma h$ with $\gamma \in G_i$ and $h \in H$.
Note that if $\gamma \in G_i$ and $h \in H$, then there exist $c_1,c_2,c_3 \in \F_{q_0}^*$ and $h_1,h_2,h_3 \in H$ satisfying
\be
c_1h_1 + c_2 h_2 + c_3 h_3 = \gamma h
\nn\ee
if and only if there exist $\bar{h}_1,\bar{h}_2,\bar{h_3} \in H$ satisfying
\be
c_1\bar{h}_1 + c_2 \bar{h}_2 + c_3 \bar{h}_3 = \gamma.
\nn\ee
Indeed we may choose $\bar{h}_1=h_1/h$, $\bar{h}_2=h_2/h$ and $\bar{h}_3=h_3/h$. Observing $\F_{q_0}^*\cdot H=H_m$ we complete the proof.
\end{proof}

Next we present two technical propositions.

\begin{proposition} \label{proposition.upper1.even}
We keep the notation and assumptions of Theorem \ref{theorem.upper1}. Assume that $0 \le i \le 2^{\ell} -1$ is an even integer. Property Pi holds if the following holds: For each $\alpha \in \F_q^*$, there exist $x_1,x_2,x_3,y_1,y_2,y_3 \in \F_q$ satisfying
\begin{align*}
x_1 + x_2 + x_3 & =  \alpha, \\
y_1 + y_2 + y_3 & =  0, \\
x_1^2 + y_1^2 & =  1, \\
x_2^2 + y_2^2 & = 1, \\
x_3^2 + y_3^2 & =  1.
\end{align*}
\end{proposition}

\begin{proof}
Let $g$ be a generator of $\F_{q^2}^*$. Put $g_1=g^{q+1}$. Note that $g_1$ is a generator of $\F_q^*$. We have
$\theta= g^{ \frac{q^2-1}{2^\ell}}$
without loss of generality as it is a primitive $2^\ell$-th root of $1$. Put
\be
\omega_1= g^{ \frac{q+1}{2^\ell}i} \;\; \mbox{and} \;\;
\omega_2= g^{ \frac{q+1}{2^\ell}i}g^{\frac{q+1}{2}}.
\nn\ee
Note that
\be \label{ep1.proposition.upper1.even}
\omega_1^{q-1}=\theta^i \;\; \mbox{and} \;\;
\omega_2^{q-1}=-\theta^i.
\ee

It is clear that $\omega_1, \omega_2 \in \F_{q^2}^*$ and $(\omega_1/\omega_2)^{q-1}=-1$, which imply that $\{\omega_1,\omega_2\}$ is a basis of $\F_{q^2}$ over $\F_q$. For $h_1, h_2, h_3\in \F_{q^2}$, let
\begin{align*}
  h_1 & = x_1 \omega_1 + y_1 \omega_2,\\
  h_2 & =x_2 \omega_1 + y_2 \omega_2,\\
  h_3&= x_3 \omega_1 + y_3 \omega_2,
\end{align*}
with $x_1,y_1,x_2,y_2,x_3,y_3 \in \F_q$. Note that if $\gamma \in \F_{q^2}^*$ with $\gamma^q = \theta^i \gamma$, then
\be \label{ep1new.proposition.upper1.even}
\gamma = \alpha \omega_1 + 0 \cdot \omega_2,
\ee
where $\alpha=\gamma/\omega_1 \in \F_q^*$. Indeed $\gamma^{q-1}=\theta^i$ by definition of Property NPi and the choice of $\gamma$.
Then using Eq. (\ref{ep1.proposition.upper1.even}) we obtain
\be
\alpha^{q-1}=\left(\frac{\gamma}{\omega_1}\right)^{q-1}=\frac{\theta^i}{\theta^i}=1,
\nn\ee
and hence $\alpha \in \F_q^*$. These arguments imply that Eq. (\ref{ep1new.proposition.upper1.even}) is the expansion of $\gamma \in \F_{q^2}^*$ with respect to the basis $\{\omega_1,\omega_2\}$ over $\F_q$.
Put $D_1=\omega_1^{q+1}$ and $D_2=\omega_2^{q+1}$. Note that
\be \label{ep2.proposition.upper1.even}
D_1=g_1^{\frac{q+1}{2^\ell}i} \;\; \mbox{and} \;\;
D_2=g_1^{\frac{q+1}{2^\ell}i}g_1^{\frac{q+1}{2}}.
\ee
We have that $h_1 \in H$ if and only if
\begin{align*}
  1 & =h_1^{q+1} \\
   & =(x_1\omega_1 + y_1 \omega_2)^q (x_1 \omega_1 + y_1 \omega_2) \\
   & =x_1^2\omega_1^{q+1} + y_1^2 \omega_2^{q+1} + x_1y_1\omega_1^q\omega_2 + x_1y_1\omega_1\omega_2^q \\
   & =x_1^2D_1 + y_1^2 D_2  + x_1y_1\omega_1 \omega_2 (\omega_1^{q-1} + \omega_2^{q-1}) \\
   & =x_1^2 D_1 + y_1^2 D_2.
\end{align*}

 Here we use the fact that $\omega_1^{q-1} + \omega_2^{q-1}=\theta^i-\theta^i=0$ using Eq. (\ref{ep1.proposition.upper1.even}).
Note that $H \subseteq H_m$. These arguments imply that Property Pi holds if the following holds: There exists $\alpha \in \F_q^*$ such that the system
\begin{align}\label{ep3.proposition.upper1.even}
x_1 + x_2  &=  \alpha,\nonumber \\
y_1 + y_2  &=  0,\nonumber \\ 
D_1 x_1^2 + D_2 y_1^2  &=  1, \\ 
D_1 x_2^2 + D_2 y_2^2 & =  1,\nonumber \\
D_1 x_3^2 + D_2 y_3^2  &=  1\nonumber 
\end{align}
is solvable with $x_1,x_2,x_3,y_1,y_2,y_3 \in \F_q$.
Note that
\be
\frac{q+1}{2^\ell} \; \mbox{is odd,  and } \; \frac{q+1}{2} \; \mbox{is even},
\nn\ee
as $\ell \ge 3$ and Eq. (\ref{Assumption3}) holds. Therefore we obtain that
\be \label{ep4.proposition.upper1.even}
D_1  \; \mbox{and} \;D_2 \;  \mbox{are nonzero squares in} \; \F_q^*,
\ee
as $i$ is even, where we use Eq. (\ref{ep2.proposition.upper1.even}).

Putting  $D_1=E_1^2$, $D_2=E_2^2$ with $E_1,E_2 \in \F_q^*$ and using the change of variables $x_1 E_1 \mapsto x_1$, $x_2 E_1 \mapsto x_2$, $x_3 E_1 \mapsto x_3$,
$y_1 E_2 \mapsto y_1$, $y_2 E_2 \mapsto y_2$  and $y_3 E_2 \mapsto y_3$ we complete the proof.
\end{proof}

\begin{proposition} \label{proposition.upper1.odd}
We keep the notation and assumptions of Theorem \ref{theorem.upper1}. Assume that $0 \le i \le 2^{\ell} -1$ is an odd integer.
Let $D \in \F_q^*$ be a non-square in $\F_q$.
Property Pi holds if the following holds: For each $\alpha \in \F_q^*$, there exist $x_1,x_2,x_3,y_1,y_2,y_3 \in \F_q$ satisfying
\begin{align*}
x_1 + x_2 + x_3 & =  \alpha, \\ 
y_1 + y_2 + y_3 & =  0, \\ 
x_1^2 + y_1^2 & =  D, \\ 
x_2^2 + y_2^2 & = D, \\ 
x_3^2 + y_3^2 & =  D.
\end{align*}
\end{proposition}
\begin{proof}
The proof is similar to the proof of Proposition \ref{proposition.upper1.even}. Using the same methods we obtain the equivalent system in Eq. (\ref{ep3.proposition.upper1.even}).
The first difference is that, instead of Eq. (\ref{ep4.proposition.upper1.even}), we have
\be
D_1  \; \mbox{and} \;D_2 \;  \mbox{are non-squares in} \; \F_q^*,
\nn\ee
as $i$ is odd.

Putting  $D_1=\frac{E_1^2}{D}$, $D_2=\frac{E_2^2}{D}$ with $E_1,E_2 \in \F_q^*$ and using the same change of variables as in the proof of Proposition \ref{proposition.upper1.even}, we complete the proof.
\end{proof}

Next we present two more technical propositions, which make the connections of Propositions \ref{proposition.upper1.even}, \ref{proposition.upper1.odd} to Weil's Sum methods we use below.

In the next proposition we exclude $\alpha \in \{-1,1\}$. We consider the case $\alpha \in \{-1,1\}$ separately in Lemma \ref{lemma.upper2.even.separate} below.

\begin{proposition} \label{proposition.upper2.even}
We keep the notation and assumptions of Theorem \ref{theorem.upper1}.
Let $\alpha \in \F_q^* \setminus \{-1,1\}$. Let $a(x),b(x),c(x) \in \F_q[x]$ be the polynomials given by
\begin{align*}
a(x) & =  2 \alpha x - \alpha^2 -1, \\ 
b(x) & =  2 \alpha x^2 + (-3 \alpha^2 -1) x + \alpha^3 + \alpha, \;\; \mbox{and} \\ 
c(x) & =  (-\alpha^2 -1) x^2 + (\alpha^3 + \alpha) x - \frac{\alpha^4}{4} - \frac{\alpha^2}{2} + \frac{3}{4}.
\end{align*}
Put
\be
\Delta(x)=b(x)^2 -4a(x)c(x) \in \F_q[x].
\nn\ee
Assume that there exists $x_1 \in \F_q$ such that
\begin{itemize}
\item[i)] $1 - x_1^2$ is a square in $\F_q$, 
\item[ii)] $a(x_1) \neq 0$, and 
\item[iii)] $\Delta(x_1)$ is a nonzero square in $\F_q$.
\end{itemize}

Then there exist $x_1,x_2,x_3,y_1,y_2,y_3 \in \F_q$ such that the system in Proposition \ref{proposition.upper1.even} holds.
\end{proposition}
\begin{proof}
We use Proposition \ref{proposition.upper1.even}. Let
\be
x_3 = \alpha - (x_1+x_2) \;\; \mbox{and} \;\;
y_3=-(y_1+y_2).
\nn\ee
Then the system in Proposition \ref{proposition.upper1.even} holds if there exist $x_1,x_2,y_1,y_2 \in \F_q$ such that the equations
\begin{align}
	x_1^2 + y_1^2 &= 1,\label{ep1.proposition.upper2.even}\\
	x_2^2 + y_2^2 &= 1,\label{ep2.proposition.upper2.even}\\
	(\alpha-(x_1+x_2))^2 + (y_1+y_2)^2 &= 1\label{ep3.proposition.upper2.even}
\end{align}
hold. From Eq. (\ref{ep3.proposition.upper2.even}) we obtain
\begin{align}
0 & =  \alpha^2 + x_1^2 + x_2^2 + 2x_1x_2 -2\alpha x_1 -2 \alpha x_2  \nonumber\\
& ~~~~ +y_1^2 + y_2^2 + 2y_1y_2 -1 \label{ep4.proposition.upper2.even}\\
&= \alpha^2 + 2x_1x_2 - 2 \alpha x_1 -2\alpha x_2 + 2y_1y_2 + 1.\nonumber
\end{align}
where we use Eq. (\ref{ep1.proposition.upper2.even}) and Eq. (\ref{ep2.proposition.upper2.even}) in the last equation.

Using assumption i), let $y_1 \in \F_q^*$ such that
\be
y_1^2=1-x_1^2.
\nn\ee
Dividing Eq. (\ref{ep4.proposition.upper2.even}) by $2y_1$ we obtain
\be \label{ep5.proposition.upper2.even}
y_2=\frac{\alpha x_1 + \alpha x_2 -x_1x_2 - \frac{\alpha^2+1}{2}}{y_1}.
\ee
Taking the square of both sides of Eq. (\ref{ep5.proposition.upper2.even}) and using Eq. (\ref{ep1.proposition.upper2.even}), Eq. (\ref{ep2.proposition.upper2.even}), we obtain
\be
(1-x_1^2)(1-x_2^2) - \left( \alpha x_1 + \alpha x_2 -x_1x_2 - \frac{\alpha^2 + 1}{2} \right)^2 = 0.
\nn\ee
By direct computation we show that the last equation is equivalent to the equation
\be \label{ep6.proposition.upper2.even}
a(x_1)x_2^2 + b(x_1)x_2 + c(x_1)=0,
\ee
where $a(x)$, $b(x)$ and $c(x)$ are the polynomials in $\F_q[x]$ defined in the statement of the proposition.

Using assumptions ii) and iii), let $x_2 \in \F_q$ be a solution of Eq. (\ref{ep6.proposition.upper2.even}). Putting the values $x_1,y_1,x_2$ in Eq. (\ref{ep5.proposition.upper2.even}) we obtain $y_2 \in \F_q$. These values of $x_1,y_1,x_2,y_2 \in \F_q$ satisfy the system of equations (\ref{ep1.proposition.upper2.even}), (\ref{ep2.proposition.upper2.even}), (\ref{ep3.proposition.upper2.even}), which completes the proof.
\end{proof}

\begin{remark} \label{remark.proposition.upper2.even}
We note that the systems in Proposition \ref{proposition.upper1.even} and \cite[Proposition II.1]{SHO2023IT} are different. Nevertheless they are related and we obtain the same polynomials in Proposition \ref{proposition.upper2.even} and  \cite[Proposition II.2]{SHO2023IT}. Note further that the assumption (i) in Proposition \ref{proposition.upper2.even} is different from the corresponding assumption in \cite[Proposition II.2]{SHO2023IT}.
\end{remark}

\begin{proposition} \label{proposition.upper2.odd}
We keep the notation and assumptions of Theorem \ref{theorem.upper1}.
Let $D \in \F_q^*$ be a non-square in $\F_q$.
Let $\alpha \in \F_q^*$. Let $a(x),b(x),c(x) \in \F_q[x]$ be the polynomials given by
\begin{align*}
a(x) & =  2 \alpha x - \alpha^2 -D, \\
b(x) & =  2 \alpha x^2 + (-3 \alpha^2 -D) x + \alpha^3 + D\alpha, \;\; \mbox{and} \\
c(x) & = (-\alpha^2 -D) x^2 + (\alpha^3 + D \alpha) x - \frac{\alpha^4}{4} - \frac{D \alpha^2}{2} + \frac{3}{4}D^2.
\end{align*}
Put
$$\Delta(x)=b(x)^2 -4a(x)c(x) \in \F_q[x].$$
Assume that there exists $x_1 \in \F_q$ such that
\begin{itemize}
\item[i)] $D - x_1^2$ is a square in $\F_q$, 
\item[ii)] $a(x_1) \neq 0$, and 
\item[iii)] $\Delta(x_1)$ is a nonzero square in $\F_q$.
\end{itemize}
Then there exist $x_1,x_2,x_3,y_1,y_2,y_3 \in \F_q$ such that the system in Proposition \ref{proposition.upper1.odd} holds.
\end{proposition}

\begin{proof}
The proof is similar to the proof of Proposition \ref{proposition.upper2.even}. A main difference is that we obtain
\be \label{ep1.proposition.upper2.odd}
y_2=\frac{\alpha x_1 + \alpha x_2 -x_1x_2 - \frac{\alpha^2+D}{2}}{y_1}
\ee
instead of Eq. (\ref{ep5.proposition.upper2.even}). Taking the square of both sides of Eq. (\ref{ep1.proposition.upper2.odd}) and using $x_1^2+y_1^2=D$, $x_2^2+y_2^2=D$, we obtain
\be
(D-x_1^2)(D-x_2^2) - \left( \alpha x_1 + \alpha x_2 -x_1x_2 - \frac{\alpha^2 + D}{2} \right)^2 = 0.
\nn\ee
The rest of the proof is similar.
\end{proof}

We need two further propositions in order to use Weil's Sum method in the theorems below. Let $\overline{\F}_q$ be an algebraic closure of $\F_q$.

\begin{proposition} \label{proposition.upper3.even}
Let $\alpha \in \F_q^* \setminus \{-1,1\}$. Let $\Delta(x) \in \F_q[x]$ be the polynomial defined in Proposition \ref{proposition.upper2.even} above. Then there is no polynomial $f(x) \in \overline{\F}_q[x]$ such that
\be
\Delta(x)=(f(x))^2.
\nn\ee
\end{proposition}
\begin{proof}
Let $\Delta_1(x)\in \F_q[x]$ be the polynomial defined as
\be
\Delta_1(x)=\frac{\Delta(x)}{4\alpha^2}.
\nn\ee
Let $A_3,A_2,A_1,A_0 \in \F_q$ be the coefficients defined as
\be
\Delta_1(x)=x^4 + A_3 X^3 + A_2 x^2 + A_1 x + A_0.
\nn\ee
Using the definition of $\Delta(x)$ in Proposition \ref{proposition.upper2.even}, by direct computation, we obtain that
\be \label{ep1.proposition.upper3.even}
A_3=\frac{-\alpha^2+1}{\alpha},
\ee
and hence $A_3 \neq 0$.

Assume the contrary that there exists $f(x) \in \overline{\F}_q[x]$ satisfying the equation in the statement of the proposition. This implies that there exist $c_1,c_0 \in \overline{\F}_q$ such that
\be \label{ep2.proposition.upper3.even}
x^4 + A_3 x^3 + A_2 x^2 + A_1 x + A_0 =(x^2+c_1 x +c_0)^2.
\ee
Using Eq. (\ref{ep2.proposition.upper3.even}) and the fact that $A_3 \neq 0$, we obtain
\be
c_1=\frac{A_0}{2}, \;\; c_0=\frac{A_1}{A_3},
\nn\ee
and
\be
A_2=c_1^2+2c_0.
\nn\ee
These imply that
\be \label{ep3.proposition.upper3.even}
A_2 - \left( \left( \frac{A_3}{2}\right)^2 + 2A_1\right)=0.
\ee
Moreover using the definition of $\Delta(x)$ in Proposition \ref{proposition.upper2.even}, by direct computation, we also obtain that
\be \label{ep4.proposition.upper3.even}
A_2 - \left( \left( \frac{A_3}{2}\right)^2 + 2A_1\right)=\frac{\alpha^2-1}{\alpha},
\ee
which is a contradiction to Eq. (\ref{ep3.proposition.upper3.even}). This completes the proof.
\end{proof}
\begin{remark} \label{remark.the.reason}
We note that we use the assumption $\alpha  \in \F_q^*\setminus \{-1,1\}$, in particular $\alpha \notin \{-1,1\}$, in the last part of the proof of Proposition \ref{proposition.upper3.even} and also in showing $A_3 \neq 0$.
\end{remark}

\begin{proposition} \label{proposition.upper3.odd}
Let $\alpha \in \F_q^*$. Let $\Delta(x) \in \F_q[x]$ be the polynomial defined in Proposition \ref{proposition.upper2.odd} above. Then there is no polynomial $f(x) \in \overline{\F}_q[x]$ such that
\be
\Delta(x)=(f(x))^2.
\nn\ee
\end{proposition}
\begin{proof}
Recall that $D \in \F_q^*$ is a non-square in $\F_q$. The proof is similar to the proof of Proposition \ref{proposition.upper3.even} above. The first difference is that, using the same methods,  we obtain
\be \label{ep1.proposition.upper3.odd}
A_3=\frac{D-\alpha^2}{\alpha},
\ee
instead of Eq. (\ref{ep1.proposition.upper3.even}). This again implies that $A_3 \neq 0$.

As in the proof of Proposition \ref{proposition.upper3.even}, we obtain that Eq. (\ref{ep3.proposition.upper3.even}) holds again.  The second difference is that, using the same methods, we obtain
\be \label{ep2.proposition.upper3.odd}
A_2 - \left( \left( \frac{A_3}{2}\right)^2 + 2A_1\right)=\frac{-D^2+D\alpha^2}{\alpha^2}
\nn\ee
instead of Eq. (\ref{ep4.proposition.upper3.even}). Combining Eq. (\ref{ep3.proposition.upper3.even}) and Eq. (\ref{ep2.proposition.upper3.odd}) we obtain a contradiction, which completes the proof.
\end{proof}

We consider the case $\alpha \in \{1,-1\}$ of Proposition \ref{proposition.upper1.even} (see also Propositions \ref{proposition.upper2.even} and \ref{proposition.upper3.even}) separately in the next lemma.

\begin{lemma} \label{lemma.upper2.even.separate}
Let $\alpha \in \{-1,1\}$. Then there exist $x_1,x_2,x_3,y_1,y_2,y_3 \in \F_q$ such that the system in Proposition \ref{proposition.upper1.even} holds.
\end{lemma}
\begin{proof}
If $\alpha=1$, then $(x_1,x_2,x_3,y_1,y_2,y_3)=(1,1,-1,0,0,0)$ is a solution. If $\alpha=-1$, then $(x_1,x_2,x_3,y_1,y_2,y_3)=(1,-1,-1,0,0,0)$ is a solution.
\end{proof}

Now we are ready to use Weil's Sum method in the following theorems.

\begin{theorem} \label{theorem.upper2.even}
We keep the notation and assumptions of Theorem \ref{theorem.upper1}. Assume that $0 \le i \le 2^{\ell} -1$ is an even integer. Then Property Pi holds.
\end{theorem}
\begin{proof}
Let $T_1, T_2, T_3$ be the subsets of $\F_q$ defined as
\be
\begin{array}{l}
T_1=\{1,-1\}, \\
T_2=\{x \in \F_q: a(x)=0\},  \\
T_3=\{x \in \F_q: \Delta(x)=0\}.
\end{array}
\nn\ee
Put $T=T_1 \cup T_2 \cup T_3$. Note that $|T_1|=2$, $|T_2|=1$, and $|T_3| \le 4$.

Let $\eta$ be the quadratic character on $\F_q$ given by
\be
\begin{array}{rcl}
\eta: \F_q & \ra & \{0,1,-1\} \\
x & \mapsto & \left\{
\begin{array}{rl}
0, & \mbox{if $x=0$}, \\
1, & \mbox{if $x \in \F_q^*$ is a square}, \\
-1, & \mbox{if $x \in \F_q^*$ is a non-square}.
\end{array}
\right.
\end{array}
\nn\ee
For $1 \le i \le 3$, let
\be
E_i=\sum_{x \in T_i} (1 + \eta(1-x^2)) (1 + \eta(\Delta(x))).
\nn\ee
Put
\be
E= \sum_{x \in T} (1 + \eta(1-x^2)) (1 + \eta(\Delta(x))),
\nn\ee
\be
N_1= \sum_{x \in \F_q \setminus T} (1 + \eta(1-x^2)) (1 + \eta(\Delta(x))),
\nn\ee
and
\be
N= \sum_{x \in \F_q } (1 + \eta(1-x^2)) (1 + \eta(\Delta(x))).
\nn\ee
It is clear that we have
\be
N_1=N-E
\nn\ee
and
\be
\begin{array}{l}
E_1 = \sum_{x \in T_1} (1 + \eta(\Delta(x))) \le 4, \\ \\
E_2 \le 2 \cdot 2 =4, \\ \\
E_3 \le 4 \cdot 2 = 8.
\end{array}
\nn\ee
Therefore we have
\be \label{ep1.theorem.upper2.even}
E \le 4 + 4 + 8 = 16.
\ee
For $x \in \F_q$ we have
\be \label{ep2.theorem.upper2.even}
\begin{array}{l}
(1 + \eta(1-x^2))(1+ \eta(\Delta(x))) \\ \\
= 1 + \eta(1-x^2) + \eta(\Delta(x)) + \eta((1-x^2)\Delta(x)).
\end{array}
\ee
It is well known that (see, for example, \cite[Theoprem 5.48]{LN}) that
\be \label{ep3.theorem.upper2.even}
\sum_{x \in \F_q} \eta(1-x^2)= -\eta(-1)=1.
\ee
Using Proposition \ref{proposition.upper3.even} and Weil's Sum (see, for example, \cite[Theoprem 5.41]{LN}) we have
\be
\left| \sum_{x \in \F_q} \eta(\Delta(x)) \right| \le 3 q^{1/2},
\ee
and
\be
\left| \sum_{x \in \F_q} \eta((1-x^2)\Delta(x)) \right| \le 5 q^{1/2}.
\ee
Using (\ref{ep2.theorem.upper2.even}) and (\ref{ep3.theorem.upper2.even}) we obtain
\be
\begin{array}{rcl}
N &= &\sum_{\x \in \F_q} 1 + \sum_{x \in \F_q} \eta(1-x^2) + \sum_{x \in \F_q} \eta(\Delta(x)) + \sum_{\F_q} \eta((1-x^2)\Delta(x)) \\ \\
& = & q+1  + \sum_{x \in \F_q} \eta(\Delta(x)) + \sum_{\F_q} \eta((1-x^2)\Delta(x)).
\end{array}
\nn\ee
These arguments imply that
\be
N \ge q+1 - 8q^{1/2}
\nn\ee
and hence
\be \label{ep4.theorem.upper2.even}
N_1 \ge q- 8q^{1/2}-15.
\ee
Note that $q-8q^{1/2}-15 > 0$ if $q > 94$. Using Proposition \ref{proposition.upper2.even} we complete the proof if $q > 94$.

The set of cardinalities $q$ such that there exists a finite $\F_q$ of odd characteristic such that $q \equiv 7 \mod 8$ and $q < 94$ is
\be
S=\{7,23,31,47,71,79,87\}.
\nn\ee
For each $q \in S$, using Magma \cite{magma} and a direct search, we show that the theorem holds. This completes the proof.
\end{proof}

\begin{theorem} \label{theorem.upper2.odd}
We keep the notation and assumptions of Theorem \ref{theorem.upper1}. Assume that $0 \le i \le 2^{\ell} -1$ is an odd integer. Then Property Pi holds.
\end{theorem}
\begin{proof}
The proof is similar to the proof of Theorem \ref{theorem.upper2.even}. Note that the polynomial $D-x^2$ has no root in $\F_q$ so that there is no need to define the set $T_1$. Let $T_2$ and $T_3$ be the sets defined similarly. Let $T=T_2 \cup T_3$. We define the sums $E$, $N_1$ and $N$ analogously. Note that we have
\be
E \le 4 + 8 =12
\nn\ee
instead of Eq. (\ref{ep1.theorem.upper2.even}). Moreover we have
\be
\sum_{x \in \F_q} \eta(D-x^2)= -\eta(-1)=1.
\nn\ee
as in Eq. (\ref{ep3.theorem.upper2.even}). Using the same arguments as in the proof of Theorem \ref{theorem.upper2.even} we obtain that
\be
N_1 \ge q- 8q^{1/2}-11.
\ee
instead of Eq. (\ref{ep4.theorem.upper2.even}). We complete the proof similarly.
\end{proof}

Combining Theorems \ref{theorem.upper1}, \ref{theorem.upper2.even} and \ref{theorem.upper2.odd}, we immediately obtain the main result of this section in the following.

\begin{corollary} \label{corollary.upper}
Let $\ell \ge 3$ be an integer. Let $\F_{q_0}$ be a finite field such that Eq. (\ref{Assumption1}) holds. Let $s \ge 1$ be an odd integer and put $q=q_0^s$. Then the covering radius of $\mathcal{C}_s(q_0)$ is at most $3$.
\end{corollary}

\section{Determining the covering radius}\label{section4}

In this section we determine the covering radius of all $\mathcal{C}_s(q_0)$ of odd characteristic which were left as an open problem in \cite{SHO2023IT}. We solve this open problem completely in this section.

Recall that $\ell \ge 3$ is an integer. Using Eq. (\ref{Assumption1}) and Eq. (\ref{Assumption2}), as in section,  we have
\be
q\equiv 2^\ell -1\!\!\!\!\! \pmod{2^{\ell +1}}.
\nn\ee
Again, let $\theta$ be a primitive $2^\ell$-th root of $1$ in $\F_{q^2}^*$. We keep the notation of $m$ and the subgroups $H$ and $H_m$ as in Section \ref{section.upper bounds}.

For index $0 \le i \le 2^\ell -1$, let Property NPi be the property defined as follows:
\begin{definition}
Property NPi: There exists $\gamma \in \F_{q^2}^*$ such that $\gamma^q=\theta^i \gamma$ and the equation
\be
h_1+h_2=\gamma
\nn\ee
is not solvable with $h_1,h_2 \in H_m$.
\end{definition}

Using the methods of \cite{SHO2023IT} and the results of Section \ref{section.upper bounds} above in this paper, we immediately obtain the following theorem. We skip its proof since it uses the same arguments in \cite{SHO2023IT}.

\begin{theorem} \label{theorem1.determineCR}
Let $\ell \ge 3 $ be an integer. Let $\F_{q_0}$ be a finite field such that Eq. (\ref{Assumption1}) holds. Let $s \ge 1$ be an odd integer and put $q=q_0^s$. Then the covering radius of $\mathcal{C}_s(q_0)$ is $3$ if and only if there exists an index $0 \le i \le 2^\ell -1$ such that Property NPi holds. Otherwise, the covering radius is $2$.
\end{theorem}

Next we present two technical propositions. Let $a$ be a generator of $\F_{q_0}^*$. Recall that $m=(q_0-1)/2$.

\begin{proposition} \label{proposition1.determineCR}
We keep the notation and assumptions of Theorem \ref{theorem1.determineCR}. Assume that $0 \le i \le 2^\ell -1$ is an even integer. Then Property NPi is equivalent to the following: There exists $\alpha \in \F_q^*$ such that the system
\begin{align*}
x_1 + x_2 & =  \alpha, \\
y_1 + y_2 & =  0, \\
x_1^2 + y_1^2 & =  a^{2j_1}, \\
x_2^2 + y_2^2 & =  a^{2j_2}
\end{align*}
is not solvable with $x_1,x_2,y_1,y_2 \in \F_q$ for any integers $0 \le j_1,j_2 \le m-1$.
\end{proposition}

\begin{proof}
The proof is similar to the proof of Proposition \ref{proposition.upper1.even} above. We use the same methods, in particular the same basis $\{\omega_1,\omega_2\}$ of $\F_{q^2}$ over $\F_q$. Let $h_1,h_2 \in \F_{q^2}$ and use their expansions
\begin{align*}
h_1 & =  x_1 \omega_1 + y_1 \omega_2, \\ 
h_2 & =  x_2 \omega_1 + y_2 \omega_2,
\end{align*}
with $x_1,y_1,x_2,y_2 \in \F_q$. As in the proof of Proposition \ref{proposition.upper1.even}, let $D_1=\omega_1^{q+1}$ and $D_2=\omega_2^{q+1}$.

Note that $h_1 \in H_m$ if and only if there exists $0 \le j_1 \le m-1$ such that (see \cite{SHO2023IT})
\be
a^{2j_1} = h_1^{q+1}.
\nn\ee
As in the proof of Proposition \ref{proposition.upper1.even} we have
\be
h_1^{q+1}=x_1^2 D_1 + y_1^2 D_2.
\nn\ee

These arguments imply that Property NPi is equivalent to the following: There exists $\alpha \in \F_q^*$ such that the system
\be \label{ep3.proposition1.determineCR}
\begin{array}{rcl}
x_1 + x_2 & = & \alpha, \vspace{1mm}\\ 
y_1 + y_2 & = & 0, \vspace{1mm}\\ 
D_1 x_1^2 + D_2 y_1^2 & = & a^{2j_1},\vspace{1mm} \\ 
D_1 x_2^2 + D_2 y_2^2 & = & a^{2j_2}
\end{array}
\ee
is not solvable with $x_1,x_2,y_1,y_2 \in \F_q$ for any integers $0 \le j_1,j_2 \le m-1$.

As in the proof of Proposition \ref{proposition.upper1.even}, we obtain that $D_1$ and $D_2$ are are nonzero squares in $\F_q^*$. Putting  $D_1=E_1^2$, $D_2^2=E_2^2$ with $E_1,E_2 \in \F_q^*$ and using the change of variables $x_1 E_1 \mapsto x_1$, $x_2 E_1 \mapsto x_2$, $y_1 E_2 \mapsto y_1$ and
$y_2 E_2 \mapsto y_2$ we complete the proof.
\end{proof}

\begin{proposition} \label{proposition2.determineCR}
We keep the notation and assumptions of Theorem \ref{theorem1.determineCR}. Assume that $0 \le i \le 2^\ell -1$ is an odd integer.
Let $D \in \F_q^*$ be a nonsquare in $\F_q$.
Then Property NPi is equivalent to the following: There exists $\alpha \in \F_q^*$ such that the system
\begin{align*}
x_1 + x_2 & =  \alpha, \\ 
y_1 + y_2 & =  0, \\ 
x_1^2 + y_1^2 & =  Da^{2j_1}, \\ 
x_2^2 + y_2^2 & =  Da^{2j_2}
\end{align*}
is not solvable with $x_1,x_2,y_1,y_2 \in \F_q$ for any integers $0 \le j_1,j_2 \le m-1$.
\end{proposition}
\begin{proof}
The proof is similar to the proof of Proposition \ref{proposition1.determineCR}. Using the same methods we obtain the equivalent system in Eq. (\ref{ep3.proposition1.determineCR}).

A main difference is that, as in the differences of the proofs of Propostions \ref{proposition.upper1.even} and \ref{proposition.upper1.odd}, instead of the fact that $D_1$ and $D_2$ are  nonzero squares in $\F_q^*$, which holds in Proposition \ref{proposition1.determineCR}, now we have
\be
D_1  \; \mbox{and} \;D_2 \;  \mbox{are non-squares in} \; \F_q^*,
\nn\ee
as $i$ is odd.

Putting  $D_1=\frac{E_1^2}{D}$, $D_2^2=\frac{E_2^2}{D}$ with $E_1,E_2 \in \F_q^*$ and using the same change of variables as in the proof of Proposition \ref{proposition1.determineCR}, we complete the proof.
\end{proof}

The proof of the following theorem uses Proposition \ref{proposition1.determineCR}.

\begin{theorem} \label{theorem2.determineCR}
We keep the notation and assumptions of Theorem \ref{theorem1.determineCR}. Assume that $0 \le i \le 2^\ell -1$ is an even integer.
Let $\alpha_1, \ldots, \alpha_m$ be an enumeration of all nonzero squares in $\F_{q_0}$. Then Property NPi is equivalent to the following: The system
\be
\begin{array}{rcl}
y_1^2 & = & x^2-\alpha_1, \\
y_2^2 & = & x^2-\alpha_2, \\
& \vdots & \\
y_m^2 & = & x^2-\alpha_m,
\end{array}
\nn\ee
is  solvable with $x,y_1,y_2, \ldots, y_m\in \F_q^*$.
\end{theorem}
\begin{proof}
Consider  the system in Proposition \ref{proposition1.determineCR} and put
\be 
x=x_2, \; y=y_2, \; x_1=x-\alpha, \; y_1=-y.
\nn\ee
Then the system in Proposition \ref{proposition1.determineCR} is equivalent to the system
\begin{align*}
(x-\alpha)^2 + y^2 & =  a^{2j_1}, \\ 
x^2+y^2 & =   a^{2j_2}.
\end{align*}
The last systems is equivalent to the system
\begin{align*}
-2x\alpha + \alpha^2& =  a^{2j_1}-a^{2j_2}, \\ 
-y^2 & =  x^2-a^{2j_2},
\end{align*}
which means
\be \label{ep1.theorem2.determineCR}
\begin{array}{rcl}
x & = & \frac{\alpha^2 + a^{2j_2} - a^{2 j_1}}{2\alpha},\vspace{1mm}\\
-y^2 & = & x^2-a^{2j_2} = \left( \frac{\alpha^2 + a^{2j_2} - a^{2 j_1}}{2\alpha} \right)^2 - a^{2j_2}.
\end{array}
\ee

Note that $-1$ is not a square in $\F_q$ as $\ell \ge 3$ and Eq. (\ref{Assumption3}) holds, which imply $q \equiv 3 \mod 4$. Hence Property NPi is equivalent to the following: There exists $\alpha \in \F_q^*$ such that
\be \label{ep2.theorem2.determineCR}
\left(\alpha^2 + a^{2j_2} - a^{2j_1}\right)^2 -4 \alpha^2 a^{2j_2}
\ee
is a square in $\F_q^*$ for any integers $0 \le j_1,j_2 \le m-1$.

We have the factorization of Eq. (\ref{ep2.theorem2.determineCR}) given by
\be \label{ep3.theorem2.determineCR}
\begin{array}{l}
\left(\alpha^2 + a^{2j_2} - a^{2j_1}\right)^2  -4\alpha^2 a^{2j_2}\vspace{1mm}\\
= \left(\alpha + a^{j_2} + a^{j_1}\right)\left(\alpha + a^{j_2} - a^{j_1}\right)\left(\alpha - a^{j_2} + a^{j_1}\right)\left(\alpha - a^{j_2} - a^{j_1}\right).
\end{array}
\ee

If $0 \le j_1=j_2 \le m-1$, then Eq. (\ref{ep3.theorem2.determineCR}) becomes
\be
\left(\alpha + 2a^{j_1}\right)\left(\alpha - 2a^{j_1}\right)= \alpha^2 -4a^{2j_1}.
\nn\ee
Note that $\{4a^{2j_1}: 0 \le j_1 \le m-1\}$ is the set $\{\alpha_1, \ldots, \alpha_m\}$ of all nonzero squares in $\F_{q_0}^*$.

The condition in the statement of the theorem is equivalent to the condition that there exists $\alpha \in \F_q^*$ such that the value in Eq. (\ref{ep3.theorem2.determineCR})
is a nonzero square in $\F_q^*$ for any integer $0 \le j_1=j_2 \le m-1$.
For $0 \le j_1, j_2 \le m-1$ with $j_1 \neq j_2$, put $u=a^{j_2} + a^{j_1}$ and $v=a^{j_2}-a^{j_1}$. Note that $u,v \in \F_q^*$. Moreover let $u^2 =\alpha_{i_0}$ and $v^2=\alpha_{i_1}$. We have
\begin{align*}
&\left(\alpha + a^{j_2} + a^{j_1}\right)\left(\alpha + a^{j_2} - a^{j_1}\right)\left(\alpha - a^{j_2} + a^{j_1}\right)\left(\alpha - a^{j_2} - a^{j_1}\right) \\ 
&= (\alpha + u)(\alpha+v)(\alpha-v)(\alpha-u) \\ 
&=\left(\alpha^2 -\alpha_{i_0}\right)\left(\alpha^2-\alpha_{i_1}\right).
\end{align*}
This implies that if both $\left(\alpha^2 -\alpha_{i_0}\right)$ and $\left(\alpha^2-\alpha_{i_1}\right)$ are nonzero squares in $\F_q^*$, then the value in (\ref{ep2.theorem2.determineCR}) is a nonzero square in $\F_q^*$. This completes the proof.
\end{proof}

The proof of the following theorem uses Proposition \ref{proposition2.determineCR}.

\begin{theorem} \label{theorem3.determineCR}
We keep the notation and assumptions of Theorem \ref{theorem1.determineCR}. Assume that $0 \le i \le 2^\ell -1$ is an odd integer.
Let $\beta_1, \ldots, \beta_m$ be an enumeration of all nonzero non-squares in $\F_{q_0}$. Then Property NPi is equivalent to the following: The system
\be
\begin{array}{rcl}
y_1^2 & = & x^2-\beta_1, \\
y_2^2 & = & x^2-\beta_2, \\
& \vdots & \\
y_m^2 & = & x^2-\beta_m,
\end{array}
\nn\ee
is  solvable with $x,y_1,y_2, \ldots, y_m\in \F_q^*$.
\end{theorem}
\begin{proof}
The proof is similar to the proof of Theorem \ref{theorem2.determineCR}. Putting $x,y$ as in Eq. (\ref{ep1.theorem2.determineCR}) we obtain
\be \label{ep1.theorem3.determineCR}
\begin{array}{rcl}
x & = & \frac{\alpha^2 + Da^{2j_2} - Da^{2 j_1}}{2\alpha}, \vspace{1mm} \\
-y^2 & = & x^2-Da^{2j_1} = \left( \frac{\alpha^2 + Da^{2j_2} - Da^{2 j_1}}{2\alpha} \right)^2 - Da^{2j_2}.
\end{array}
\ee
instead of Eq. (\ref{ep1.theorem2.determineCR}).

Let $E \in \F_{q^2} \setminus \F_q$ such that $E^2=D$. Using the same methods as in the proof of Theorem \ref{theorem2.determineCR} we obtain
\be \label{ep3.theorem3.determineCR}
\begin{array}{l}
\left(\alpha^2 + Da^{2j_2} - Da^{2j_1}\right)^2  -4D\alpha^2 a^{2j_2}\vspace{1mm} \\
= \left(\alpha + Ea^{j_2} + Ea^{j_1}\right)\left(\alpha + Ea^{j_2} - Ea^{j_1}\right)\left(\alpha - Ea^{j_2} + Ea^{j_1}\right)\left(\alpha - Ea^{j_2} - Ea^{j_1}\right).
\end{array}
\ee
instead of Eq. (\ref{ep3.theorem2.determineCR}).

If $0 \le j_1=j_2 \le m-1$, then Eq. (\ref{ep3.theorem3.determineCR}) becomes
\be
\left(\alpha + 2E a^{j_1}\right)\left(\alpha - 2 E a^{j_1}\right)= \alpha^2 -4Da^{2j_1}.
\nn\ee
Note that $\{4Da^{2j_1}: 0 \le j_1 \le m-1\}$ is the set $\{\beta_1, \ldots, \beta_m\}$ of all nonzero non-squares in $\F_{q_0}^*$.

For $0 \le j_1, j_2 \le m-1$ with $j_1 \neq j_2$, put $u=E\left(a^{j_2} + a^{j_1}\right)$ and $v=E\left(a^{j_2}-a^{j_1}\right)$. Note that $u^2,v^2 \in \{\beta_1, \ldots, \beta_m\}$ and $u^2 \neq v^2$. Put $u^2 =\beta_{i_0}$ and $v^2=\beta_{i_1}$. We have
\be
\begin{array}{l}
\left(\alpha + Ea^{j_2} + Ea^{j_1}\right)\left(\alpha + Ea^{j_2} - Ea^{j_1}\right)\left(\alpha - Ea^{j_2} + Ea^{j_1}\right)\left(\alpha - Ea^{j_2} - Ea^{j_1}\right) \vspace{1mm} \\
=\left(\alpha^2 -\beta_{i_0}\right)\left(\alpha^2-\beta_{i_1}\right).
\end{array}
\nn\ee
We complete the proof using similar arguments as in the proof of Theorem \ref{theorem2.determineCR}. Note that we apply Proposition \ref{proposition2.determineCR} instead of Proposition \ref{proposition1.determineCR}.
\end{proof}

Next we immediately determine the covering radius of $\mathcal{C}_s(q_0)$ if $s=1$.

\begin{corollary} \label{corollary1.determineCR}
We keep the notation and assumptions of Theorem \ref{theorem1.determineCR}. Then the covering radius of $\mathcal{C}_s(q_0)$ is $2$ if $s=1$.
\end{corollary}
\begin{proof}
Let $s=1$, which means $\F_{q_0}=\F_q$. Using Theorem \ref{theorem1.determineCR}, it is enough to show that neither the system in Theorem \ref{theorem2.determineCR} is solvable nor the system in Theorem \ref{theorem3.determineCR}.

First consider the system in Theorem \ref{theorem2.determineCR}. Assume the contrary that $x,y_1, \ldots, y_m \in \F_q^*$ is a solution. Then $x^2 \in \{\alpha_1, \ldots, \alpha_m\}$ as $\F_q=\F_{q_0}$. If $x^2=\alpha_{i_0}$ with $1 \le i_0 \le m$, then $y_{i_0}=0$, which is a contradiction.

Next consider the system in Theorem \ref{theorem3.determineCR}. Assume the contrary that $x,y_1, \ldots, y_m \in \F_q^*$ is a solution. Then
the set $\{y_1^2, \ldots, y_m^2\}$ has cardinality $m$ as the set $\{\beta_1, \ldots, \beta_m\}$ has cardinality $m$. This implies that the set $\{y_1^2, \ldots, y_m^2\}$
is exactly the set $\{\alpha_1, \ldots, \alpha_m\}$, where we use the fact that $\F_q=\F_{q_0}$. Therefore $x^2 \in \{y_1^2, \ldots, y_m^2\}$ and hence there exists $1 \le i_0 \le m$ such that $\beta_{i_0}=0$, which is a contradiction. This completes the proof.
\end{proof}

Recall that $\ell \ge 3$ is an integer, $\F_{q_0}$ is a finite field satisfying Eq. (\ref{Assumption1}) and $m=(q_0-1)/2$. Let $s^*$ be the smallest odd integer such that
\be \label{s*.determinedCR}
q_0^{s^*}+1 -2 \left( 1 + 2^{m-1}(m-2)\right)q_0^{s^*/2} > 2^m.
\ee

Now we are ready to prove the main result of this section.

\begin{theorem} \label{theorem.main.determineCR}
We keep the notation and assumptions of Theorem \ref{theorem1.determineCR}. Let $s^*$ be the (positive) odd integer defined in Eq. (\ref{s*.determinedCR}).

If $s \ge s^*$ is an odd integer, then the covering radius of $\mathcal{C}_s(q_0)$ is $3$.
\end{theorem}
\begin{proof}
We use some technical but important results using the theory of algebraic curves over finite fields. We refer to \cite{G-Sti} and Appendix C of \cite{SHO2023IT} for further background on algebraic curves.

We will show that if $s\ge s^*$ is an odd integer, then Property NPi holds for each even integer $i$ in the range $0 \le i \le 2^\ell-1$, which completes the proof using Theorem \ref{theorem1.determineCR}.

Let $\alpha_1, \ldots, \alpha_m$ be an enumeration of all nonzero squares in $\F_{q_0}$. Let $\chi$ be the fibre product of projective lines over $\F_q$ given by
\be
\begin{array}{rcl}
\chi & : & \left\{
\begin{array}{rcl}
y_1^2 &=& x^2 - \alpha_1, \\
y_2^2 &=& x^2 - \alpha_2, \\
& \vdots & \\
y_m^2 &=& x^2 - \alpha_m.
\end{array}
 \right.
\end{array}
\nn\ee

Let $P_\infty$ be the pole of $x$ in $\F_q(x)$. For $u \in \F_q$, let $P_u$ be the zero of $(x-u)$ in $\F_q(x)$.

Using \cite[Lemma IV.1]{SHO2023IT} and the methods in the proof of \cite[Theorem IV.1]{SHO2023IT}, we conclude that there is no $\F_q$-rational point of $\chi$ over $P_u$ if $u \in \{0,\alpha_1, \ldots, \alpha_m\}$. Similarly there are exactly $2^m$ $\F_q$-rational points of $\chi$ over $P_\infty$. Moreover if $u \in \F_q \setminus\{0,\alpha_1, \ldots, \alpha_m\}$, then there are either no or exactly $2^m$ $\F_q$-rational points of $\chi$ over $P_u$. Furthermore, if $u \in \F_q \setminus\{0,\alpha_1, \ldots, \alpha_m\}$, then there exist exactly $2^m$ $\F_q$-rational points of $\chi$ over $P_u$ if and only if $(u^2 -\alpha_i)$ is a nonzero square in $\F_q$ for each $1 \le i \le m$.

Let $N(s)$ denote the number of $u \in \F_q$ such that $(u^2 -\alpha_i)$ is a nonzero square in $\F_q$ for each $1 \le i \le m$.
Let $N(\chi;s)$ denote the number of $\F_q$-rational points of $\chi$.
The arguments above imply that we have
\be \label{ep1.theorem.main.determineCR}
N(\chi;s)=2^m+2^mN(s).
\ee

Using Theorem \ref{theorem2.determineCR}, it is enough to show that $N(s) >0$ if $s$ is an odd integer with $s \ge s^*$. Combining this with Eq. (\ref{ep1.theorem.main.determineCR}) we conclude that it is enough to show that
\be \label{ep2.theorem.main.determineCR}
N(\chi;s)>2^m
\ee
if $s$ is an odd integer with $s \ge s^*$.

Note that the genus of $\chi$ is $1+2^{m-1}(m-2)$, which is computed in the proof of \cite[Theorem IV.1]{SHO2023IT}. Hence using Hasse-Weil Inequality (see, for example, \cite{G-Sti}), we obtain that
\be \label{ep3.theorem.main.determineCR}
 N(\chi;s) \ge q_0^s + 1 -2\left(1 + 2^{m-1}(m-2)\right) q_0^{s/2}.
\ee
Using the methods of in the proof of \cite[Theorem IV.4]{SHO2023IT} we have that
\be \label{ep4.theorem.main.determineCR}
 q_0^s + 1 -2\left(1 + 2^{m-1}(m-2)\right) q_0^{s/2} > 2^m
 \ee
if $s$ is an odd integer with $s \ge s^*$. Combining Eq. (\ref{ep2.theorem.main.determineCR}), (\ref{ep3.theorem.main.determineCR}) and (\ref{ep4.theorem.main.determineCR}) we complete the proof.
\end{proof}

The following definition is crucial for the characterization of the covering radius of the codes in this paper.

\begin{definition} \label{definition.Iq0}
Let $\ell \ge 3 $ be an integer. Let $\F_{q_0}$ be a finite field such that Eq. (\ref{Assumption1}) holds. Let $I(q_0)$ be the set consisting of odd integers $s \ge 1$ such that the covering radius of $\mathcal{C}_s(q_0)$ is $3$.

It follows from Theorem \ref{theorem1.determineCR} that if $s \ge 1$ is an odd integer and $s \not \in I(q_0)$, then the covering radius of $\mathcal{C}_s(q_0)$ is $2$.
\end{definition}

Moreover it follows from Corollary \ref{corollary1.determineCR} that $1 \not \in I(q_0)$ for any finite field $\F_{q_0}$ such that Eq. (\ref{Assumption1}) holds, where $\ell \ge 3 $ is an integer.
The following is an immediate corollary of Theorem \ref{theorem.main.determineCR}.

\begin{corollary} \label{corrollary.Iq0.sturates}
Let $\ell \ge 3 $ be an integer. Let $\F_{q_0}$ be a finite field such that Eq. (\ref{Assumption1}) holds. Let $s^*$ be the (positive) odd integer defined in Eq. (\ref{s*.determinedCR}).
If $s \ge s^*$ is an odd integer, then $s \in I(q_0)$.
\end{corollary}

Next we present a useful property of $I(q_0)$.

\begin{proposition} \label{proposition.property}
Let $\ell \ge 3 $ be an integer. Let $\F_{q_0}$ be a finite field such that Eq. (\ref{Assumption1}) holds. Let $t \ge 1$ be an odd integer.
If $s \in I(q_0)$, then for the odd integer $st$ we have $st \in I(q_0)$.
\end{proposition}
\begin{proof}
Let $m=(q_0-1)/2$. Let $\alpha_1, \ldots, \alpha_m$ be an enumeration of the nonzero squares in $\F_{q_0}$. Let $\beta_1, \ldots, \beta_m$ be an enumeration of the nonzero
non-squares in $\F_{q_0}$. Using Theorems \ref{theorem1.determineCR}, \ref{theorem2.determineCR} and \ref{theorem3.determineCR}, we obtain that at least one of the two following cases hold:
\begin{itemize}
\item[Case 1:] There exist $x,y_1, \ldots, y_m \in \F_{q_0^s}^*$ such that the system in Theorem \ref{theorem2.determineCR} holds.
\item[Case 2:] There exist $x,y_1, \ldots, y_m \in \F_{q_0^s}^*$ such that the system in Theorem \ref{theorem3.determineCR} holds.
\end{itemize}
As $\F_{q_0}^s \subset \F_{q_0}^{st}$, in each of the cases above we also obtain solutions of the corresponding systems with $x,y_1, \ldots, y_m  \in \F_{q_0^{st}}^*$. This completes the proof.
\end{proof}

Next we explicitly determine $I(q_0)$ for the smallest finite fields $\F_{q_0}$ which are not in $S$, so that which are not considered in \cite{SHO2023IT}.

\begin{example} \label{example q0 is 7}
Let $\ell=3$ and $q_0=7$. Using Magma we obtain that the odd integer $s^*$ defined in (\ref{s*.determinedCR}) becomes $s^*=3$ in this example. Hence we determine that
\be
I(7)=\{s \ge 3: \; \mbox{$s$ is an odd integer}\}.
\nn\ee
\end{example}

\begin{example} \label{example q0 is 23}
Let $\ell=3$ and $q_0=23$. Using Magma we obtain that the odd integer $s^*$ defined in (\ref{s*.determinedCR}) becomes $s^*=7$ in this example.
Moreover using Magma and Theorems \ref{theorem2.determineCR}, \ref{theorem3.determineCR} we further obtain that $3 \not \in I(23)$ and $5 \in I(23)$.
We note that the computation for the statement  $5 \in I(23)$ is quite time consuming.
Hence we determine that
\be
I(23)=\{s \ge 5: \; \mbox{$s$ is an odd integer}\}.
\nn\ee
\end{example}

\begin{example} \label{example q0 is 31}
Let $\ell=5$ and $q_0=31$. Using Magma we obtain that the odd integer $s^*$ defined in (\ref{s*.determinedCR}) becomes $s^*=9$ in this example.
Moreover using Magma and Theorems \ref{theorem2.determineCR}, \ref{theorem3.determineCR} we further obtain that $3 \not \in I(31)$ and $5 \in I(31)$.
We note that the computation for the statement deciding if $7$ is  in $I(31)$  or not takes too long and we had to stop after some time. Hence we do not know if $7$ is  in $I(31)$  or not.
Hence we determine that
\be
I(31)=\{s \ge 5: \; \mbox{$s$ is an odd integer}\} \; \mbox{or} \;
I(31)=\{s \ge 9: \; \mbox{$s$ is an odd integer}\} \cup \{5\}.
\nn\ee
\end{example}

\begin{example} \label{example q0 is 47}
Let $\ell=4$ and $q_0=47$. Using Magma we obtain that the odd integer $s^*$ defined in (\ref{s*.determinedCR}) becomes $s^*=11$ in this example.
Moreover using Magma and Theorems \ref{theorem2.determineCR}, \ref{theorem3.determineCR} we further obtain that $3 \not \in I(47)$ and $5 \not \in I(47)$.
We note that the computation for the statement  $5 \not \in I(47)$ is extremely time consuming.
We note that the computation for the statement deciding if $7$ is  in $I(47)$  or not takes too long and we had to stop after some time.
Similarly the computation for the statement deciding if $9$ is  in $I(47)$  or not takes too long and we had to stop after some time.
Hence we do not know if $7$ is  in $I(47)$  or not. Similarly we do not know if $9$ is  in $I(47)$  or not.
Hence we determine that
\be
I(47)=\{s \ge 11: \; \mbox{$s$ is an odd integer}\} \cup J,  \; \mbox{where} \;
J \subseteq \{7,9\}.
\nn\ee
Note that the options for $J$ are $\emptyset$, $\{7\}$, $\{9\}$, and $\{7,9\}$.
\end{example}

\section{Covering radius of twisted half generalized Zetterberg codes in odd characteristic}\label{section.half twisted half}

In this section, we introduce the concept of twisted half generalized Zetterberg codes in odd characteristic. We also show that
their covering radii are the same as that of the (full) generalized Zettenberg codes, that we consider in Sections \ref{section.upper bounds} and \ref{section4}.

Let $\ell \ge 3 $ be an integer. Let $\F_{q_0}$ be a finite field such that (\ref{Assumption1}) holds.  For an integer $s \ge 1$ let $q=q_0^s$ and $n=(q_0^s+1)/2$.

\begin{definition} \label{definition.twisted.half.zettenberg}
Under the notations as above, assume that $q=q_0^s\in S_\ell$, i.e., $q \equiv 2^{\ell}-1 \pmod{2^{\ell+1}}$. Let $H$ be the subgroup of $\F_{q^2}^*$ with $|H|=q+1$. Let $H_\ell$ be the subgroup of $\F_{q^2}^*$ with $|H_\ell|=\frac{q+1}{2^\ell}$. Let $\theta\in\mathbb{F}_{q^2}$ be a primitive $2^\ell$-th root of 1.
	Note that $\gcd\left(\frac{q+1}{2^\ell}, 2^\ell\right)=1$ and $\theta^{2^{\ell-1}}=-1$, hence $H$ can be decomposed
	into the following disjoint union:
	\be
	H=\bigcup_{i=0}^{2^\ell-1}\theta^i H_\ell=\left( \bigcup_{i=0}^{2^{\ell-1}-1}\theta^i H_\ell\right) \bigcup
	-\left(\bigcup_{i=0}^{2^{\ell-1}-1}\theta^i H_\ell\right)
	\stackrel{\triangle}{=} \mathcal{H}\cup -\mathcal{H}.
	\nn\ee
	We define the {\it twisted half generalized Zettenberg code $\mathcal{C}^{t}_s(q_0)$ of length $n=(q+1)/2$ over $\F_{q_0}$}
	as the linear code over
	$\F_{q_0}$ with the parity check matrix
	\be
	P_t=\left[
	\begin{array}{cccccccccc}
	\beta_1~\beta_2~\cdots~\beta_{n}
	\end{array}
	\right],
	\nn\ee
	where $\beta_1,\beta_2,\ldots,\beta_{n}$ are all elements of $\mathcal{H}.$
\end{definition}

In the next two theorems we determine the covering radius of these codes. We also characterize when these codes are quasi-perfect.

\begin{theorem} \label{CR GS constacyclic}
Let $\ell \ge 3 $ be an integer. Let $\F_{q_0}$ be a finite field such that (\ref{Assumption1}) holds.  Let $s\ge 1$ be an odd integer and put $n=(q_0^s+1)/2$. Let $\mathcal{C}^{t}_s(q_0)$ be the twisted half generalized Zettenberg code of length $n$ over $\F_{q_0}$ is defined in Definition \ref{definition.twisted.half.zettenberg}. Then the covering radius of $\mathcal{C}^{t}_s(q_0)$ is either $2$ or $3$, and its covering radius is $3$ if and only if $s \in I(q_0)$. Moreover $\mathcal{C}^{t}_s(q_0)$ is quasi-perfect if and only if $s \not \in I(q_0)$.
\end{theorem}

\begin{proof}
Firstly, we discuss the parameters of the twisted half generalized Zettenberg code $\mathcal{C}^{t}_s(q_0)$. From \cite[Lemma A.1]{SHO2023IT}, the Zetterberg code $\mathcal{C}_s(q_0)$ has dimension $q_0^s+1-2s$. Through this, it is not difficult to obtain that the twisted half generalized Zettenberg code $\mathcal{C}^{t}_s(q_0)$ of length $n=(q+1)/2$ also has dimension $n-2s$. 
Next, we consider the minimum distance $d$ of the twisted half generalized Zettenberg code $\mathcal{C}^{t}_s(q_0)$. Since all columns in $P_t$ are nonzero it follows that $d\geq 2$.
Let $h_\ell\in H_\ell$ be a generator of $H_\ell$. Note that $\gcd\left(q_0^s+1, q_0-1\right)=2$ and hence 
$$H\cap \F_{q_0}=\{1,-1\}.$$
It folows from Definition \ref{definition.twisted.half.zettenberg} that $1\in \mathcal{H}$ and $-1\notin \mathcal{H}$. If there exists a a codeword of weight $2$, then there are two distinct elements $\beta_1,\beta_2\in \mathcal{H}$ and $c_1,c_2\in \F_{q_0}^*$ such that $c_1\beta_1=c_2\beta_2$ and $\beta_1\neq \beta_2$. It turns out that $\frac{\beta_1}{\beta_2}=\frac{c_1}{c_1}\in H\cap \F_{q_0}=\{1,-1\}$.
Since $\beta_1\neq \beta_2$, $\beta_1=-\beta_2$. This is impossible because $H$ is the disjoint union of $\mathcal{H}$ and $-\mathcal{H}$.
Hence, $d\geq 3$. By the well-known sphere-packing bound \cite{Mac}, we obtain that $d\leq 4$. Therefore, the twisted half generalized Zettenberg code $\mathcal{C}^{t}_s(q_0)$ has the parameters
$[n, n-2s, 3 \leq d \leq 4]$.

Let $\mathcal{C}_s(q_0)$ be the generalized Zetterberg code of length $q_0^s+1$ over $\F_{q_0}$. Using the results above we obtain that the covering radius of $\mathcal{C}_s(q_0)$ is either $2$ or $3$, and its covering radius is $3$ if and only if $s \in I(q_0)$.
Using the same methods in the proof of \cite[Theorem VI.4]{SHO2023IT}, we obtain that covering radius of $\mathcal{C}^{t}_s(q_0)$ is the same as the covering radius of $\mathcal{C}_s(q_0)$.

As the minimum distance $d$ of $\mathcal{C}^{t}_s(q_0)$ satisfies that $3 \le d \le 4$, we obtain that $\lfloor (d-1)/2 \rfloor=1$ and hence $\mathcal{C}^{t}_s(q_0)$ is quasi-perfect if and only if $s \not \in I(q_0)$.
\end{proof}

Let $\ell \ge 3 $ be an integer. Let $\F_{q_0}$ be a finite field such that (\ref{Assumption1}) holds. Using Theorem \ref{CR GS constacyclic} we obtain many examples of
quasi-perfect codes. Indeed for each such $\F_{q_0}$ by putting $s=1$ we obtain a quasi-perfect $[(q_0+1)/2,(q_0-3)/2, 3 \le d \le 4]$ code over $\F_{q_0}$.

Moreover using Examples
\ref{example q0 is 23}, \ref{example q0 is 31}, \ref{example q0 is 47}, if $s=3$ and $q_0 \in \{23,31,47\}$, then we obtain a quasi-perfect $[(q_0^3+1)/2,(q_0^3-11)/2, 3 \le d \le 4]$ code over $\F_{q_0}$.

Furthermore using Example \ref{example q0 is 47}, putting $s=5$ and $q_0=47$, we obtain a quasi-perfect $[(q_0^5+1)/2,(q_0^5-19)/2, 3 \le d \le 4]$ code over $\F_{q_0}$.

It seems that if $q_0$ is large, then there are further values of odd integers $1 < s < s^*$ such that $s \not \in I(q_0)$ and hence quasi-perfect
$[(q_0^s+1)/2,(q_0^s+1)/2-2s, 3 \le d \le 4]$ codes over $\F_{q_0}$, where $\F_{q_0}$ is a finite field such that (\ref{Assumption1}) holds for an integer $\ell \ge 3$.

\section{Concluding remarks}\label{section6}
In this paper, we have determined the covering radius of generalized Zetterberg codes for $q_0^s \equiv 7 \pmod{8}$, which therefore solves an open problem proposed in \cite{SHO2023IT}. Togehter with the results of \cite{SHO2023IT}, we have determined the covering radius of all generalized Zetterberg codes in odd characteristic. We also introduce the concept of twisted half generalized Zetterberg codes of length $\frac{q_0^s+1}{2}$, and show the same results hold for them. As a result, we obtain some quasi-perfect codes.

It would be interesting but hard to consider the covering radius of Zetterberg type codes in even characteristic in a future work.

\section*{Acknowledgments}
The research of Shitao Li and Minjia Shi is supported by the National Natural Science Foundation of China under Grant 12471490. The research of Tor Helleseth was supported by the Research Council of Norway under Grant 247742/O70. The research of Ferruh \"{O}zbudak is supported by T\"{U}B\.{I}TAK under Grant 223N065.

\section*{Declarations}

\noindent\textbf{Data availability} No data are generated or analyzed during this study.  \\

\noindent\textbf{Conflict of Interest} The authors declare that there is no possible conflict of interest.

\end{sloppypar}
\end{document}